\documentclass{article}
\usepackage{spconf,amsmath,graphicx}

\input{mysymbol.sty}

\usepackage[utf8]{inputenc} 
\usepackage[T1]{fontenc}    
\usepackage{hyperref}       
\usepackage{url}            
\usepackage{booktabs}       
\usepackage{amsfonts}       
\usepackage{nicefrac}       
\usepackage{microtype}      
\usepackage{lipsum}
\usepackage{amsmath}
\usepackage{amssymb}        
\usepackage{amsthm} 
\theoremstyle{plain}
\newtheorem{thm}{Theorem} 

\theoremstyle{definition}
\newtheorem{remark}{Remark}

\usepackage{multirow}
\usepackage{verbatim}
\usepackage{empheq,float}
\usepackage{graphicx,caption,subcaption,balance}
\usepackage{stfloats}
\usepackage[dvipsnames]{xcolor}

\newcommand{\sn}{\smallskip\noindent}

\usepackage{epstopdf}
\usepackage{algorithm}
\usepackage{algpseudocode}
\usepackage{mathrsfs}
\usepackage{psfrag}

\title{Blind Polynomial Regression}
\name{Alberto Natali and Geert Leus\thanks{This work is part of the GraSPA project (project 19497 within the TTW OTP programme), which is financed by the Netherlands Organization for Scientific Research (NWO). E-mails: \{a.natali;  g.j.t.leus\}@tudelft.nl; }}
\address{Faculty of Electrical Engineering, Mathematics and Computer Science\\ Delft University of Technology, Delft, The Netherlands}

\begin{document}
\ninept
\setlength{\abovedisplayskip}{5pt}
\setlength{\belowdisplayskip}{5pt}
\maketitle
\begin{abstract}
Fitting a polynomial to observed data is an ubiquitous task in many signal processing and machine learning tasks, such as interpolation and prediction. In that context, input and output pairs are available and the goal is to find the coefficients of the polynomial. However, in many applications, the input may be partially known or not known at all, rendering conventional regression approaches not applicable. In this paper, we formally state the (potentially partial) blind regression problem, illustrate some of its theoretical properties, and propose algorithmic approaches to solve it. As a case-study, we apply our methods to a jitter-correction problem and corroborate its performance. 
\end{abstract}
\vspace{-.2cm}
\section{Introduction}

\vspace{-.3cm}

    Parameter estimation is a classical task in signal processing~\cite{kay1993fundamentals}, where the objective is to estimate the parameters of a \textit{hypothesized} model that best explains the relationship between \textit{observed} input-output data pairs. When the input is not directly observable, standard estimation approaches such as  the ordinary least squares, cannot be directly applied, and new estimation techniques are required. 
    
    Motivated by a recent problem arising in graph signal processing~\cite{shuman2013emerging} in the context of graph learning~\cite{natali2022dual}, in this work we consider a polynomial regression model, i.e., where the output variable is expressed as the linear combination of powers of the input variable, with possible additive noise. As we will see, even though the problem can be approached through the lens of  matrix factorization~\cite{fu2019nonnegative} with a Vandermonde factor, a rather natural interpretation is given by that of \textit{sampling multiple unknown polynomials at the same unknown locations}. Specifically, we assume to observe a certain number of sampled polynomials, whose coefficients and sample locations are unknown or noisy, and the goal is to jointly recover them from the available observations.

With reference to prior work,   one of the earliest works in the context of sampling at unknown locations is that of~\cite{marziliano2000reconstruction}, which considers discrete time bandlimited signals where only a subset of its values are known in correspondence to  an ordered unknown subset of the index set. The resulting combinatorial optimization problem is then approached with an exhaustive search and two heuristics methods. The work in~\cite{browning2007approximating} considers the continuous case, and proposes an alternating least squares method that converges to a local minimum. More similar to our work is that of~\cite{elhami2018sampling}, which considers sampled polynomials (and bandlimited functions), and gives the conditions for the unique identifiability of the sampled polynomial. Specifically, it is shown that unique identifiability is achieved by constraining the sampling locations to be described by rational functions.

 Our work can be considered complementary to that of~\cite{elhami2018sampling} with some distinctions;  first, we do not constrain the sample locations, and assume to observe different sampled polynomials all sampled at the same \textit{unknown} locations. This  implicitly restricts the feasible sample locations. Then, we study which are the model ambiguities, representing an equivalence class and show that for some problems, namely the one in~\cite{natali2022dual}, every element of this class (which is a solution for the problem) is a good solution. Finally, we propose two algorithmic methods to solve it: the first based on exhaustive search, the second based on a subspace approach. Numerical simulations in the context of sampling-jitter corroborate our theory.

\section{Problem Definition}
\vspace{-3mm}
\label{sec:preliminaries}
Consider a simple polynomial regression model~\cite{hastie2009elements} of the form:
\begin{align}
\label{eq:polynomial}
    y&= w_0 + w_1x+\ldots+w_{K-1}x^{K-1} + e \\
    &= \bbw^\top \bbv(x) +e,
\end{align}
where $y$ is the independent variable (observation), $x$ is the dependent variable (\textit{regressor}), $\bbw:= [w_0, \ldots, w_{K-1}]^\top$ are the  coefficients associated to the different polynomial degrees and $\bbv(x):= [1, x, \ldots, x^{K-1}]^\top$ is the Vandermonde vector with parameter $x \in \reals$; $e \sim \ccalN(0, \sigma^2)$ is a noise term.

\noindent
Although polynomial regression fits a non linear model to the (available) data, from an estimation point of view it is linear, meaning that for a fixed input variable $x$, the dependence of the output variable $y$ on the weights in $\bbw$ is linear. This  \textit{data-availability}  scenario is often encountered in signal processing in the context (for instance) of signal interpolation and prediction, where input-output pairs $\{(x_i, y_i)\}_{i=0}^{N-1}$ are available by sampling an unknown function in $N$ distinct points, and the goal is to estimate the coefficient vector $\bbw$ which, for a specified $K$, best represents the unknown polynomial relationship between $x$ and $y$. 

\sn
Expressing~\eqref{eq:polynomial} for all the input-output pairs leads to the matrix-vector form\footnote{We will neglect the error term in the rest of the exposition.}:
\vspace{-.2cm}
\begin{align}
\label{eq:polynomial-regression}
    \bby= \bbV(\bbx)\bbw
\end{align}
\sloppy where $\bbx:=[x_0, \ldots, x_{N-1}]^\top$ is the vector of inputs, $\bby=[y_0, \ldots , y_{N-1}]^\top$ is the vector of outputs, and where $\bbV(\bbx)$ is the $N \times K$ Vandermonde matrix, defined as:
\begin{align}
\label{eq:vandermonde}
    \bbV(\bbx):= \begin{bmatrix}
    1 & x_0 & \cdots & x_0^{K-1}\\
    1 &  x_1 & \cdots & x_1^{K-1}\\
    \vdots &  \vdots & \cdots & \vdots\\
    1 &  x_{N-1} & \cdots & x_{N-1}^{K-1}
    \end{bmatrix}
\end{align}
The unique  solution of~\eqref{eq:polynomial-regression}, guaranteed when all $x_i$'s are different and $N>K$, can be found as $\hat{\bbw}=(\bbV\bbV)^{-1}\bbV^\top \bby $. As a short-hand notation, we will use $\bbV$ to refer to the Vandermonde matrix with parameter $\bbx$ in~\eqref{eq:vandermonde}; when convenient for clarity of exposition, we will explicitly write $\bbV(\bbx)$.

\sn
However, it is often the case that variable $x$ is not perfectly known or not known at all. For instance, in analog-to-discrete conversion of signals, undesired jitter on the clock signal leads to a clock drift with respect to the reference clock and a re-synchronization scheme is necessary. 
For this reason, in this work we consider the scenario in which  not only the linear weights $\bbw$ are unknown, but also the vector  $\bbx$. In order to overcome unique identifiability and solvability issues, we assume to obtain $L\geq K$ observation vectors $\{\bby_i\}_{i=0}^{L-1}$ of~\eqref{eq:polynomial-regression} associated to different \textit{unknown} coefficient vectors $\{\bbw_i\}_{i=0}^{L-1}$ yet having the same \textit{unknown} vector $\bbx$, that is,:
\begin{align}
\label{eq:matrix-factorization}
    \bbY= \bbV(\bbx)\bbW 
\end{align}
where $\bbY:=[\bby_0, \ldots, \bby_{L-1}]$  and $\bbW:=[\bbw_0, \ldots, \bbw_{L-1}]$. Such a model can be motivated by the fact that the $L$ observation channels are controlled by the same clock and hence have the same potential clock jitter. Based on~\eqref{eq:matrix-factorization}, the problem statement can be formalized as follows: 

\sn
\textbf{Problem Statement.} \textit{Given the matrix $\bbY \in \reals^{N \times L}$, recover the input vector $\bbx \in \reals^N$ and the coefficient matrix $\bbW \in \reals^{K \times L}$ such that~\eqref{eq:matrix-factorization} holds.}

\sn
Although the problem can be approached from a pure algebraic point of view as a structured matrix factorization,  a pleasing geometrical interpretation of~\eqref{eq:matrix-factorization} is given in Fig.~\ref{fig:idea}. Each vector $\bby_0, \ldots, \bby_{L-1}$ can be interpreted as function values obtained by sampling $L$ distinct polynomials $y_0(x), \ldots, y_{L-1}(x)$, all with  degree $K-1$, in the same $N$ unknown locations $x_0, \ldots, x_{N-1}$. The goal is to recover the original locations (and polynomial coefficients) from the available sampled function values (filled points on the right of the figure).

\begin{figure}
    \centering
   \includegraphics[width=\columnwidth, keepaspectratio=true,  trim=1.8cm 0 1.9cm 0.4cm, clip=true]{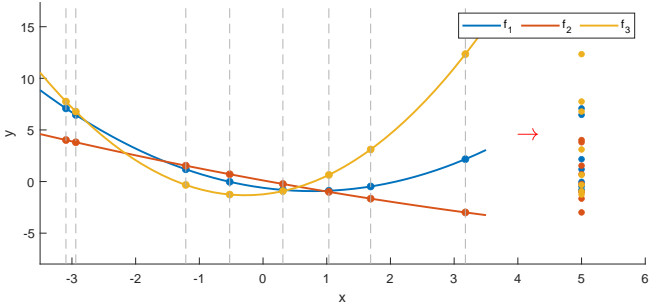}
    \vspace{-.7cm}
    \caption{Illustration of the problem: a number of functions ($L=3$) are sampled (filled points) in $N=8$ locations (gray dotted line) The problem is to recover the sampling locations and the polynomial coefficients with the only knowledge of the function values (shown in the right in the figure). }
    \label{fig:idea}
\end{figure}

\vspace{-.2cm}
\begin{remark}
\label{remark:doa}
At  first glance, model~\eqref{eq:matrix-factorization} shares similarities with fundamental array signal processing models encountered for instance in direction of arrival (DoA) estimation~\cite{chung2014doa}. However, the geometry of our problem is different (transposed) and complicates its analysis in that  ready-to-use algorithms, such as MUSIC~\cite{schmidt1986multiple}, are not directly applicable; see also Remark~\ref{remark:music}.
\end{remark}


\section{Proposed Problem Solutions}
\vspace{-.3cm}

  Before delving into the algorithmic framework, we demand ourselves which are the solution pairs $(\bbx,\bbW)$ that can be considered a ``valid solution'' for problem~\eqref{eq:matrix-factorization}. Indeed notice how~\eqref{eq:matrix-factorization} is not free of model ambiguities, and different pairs $(\bbx,\bbW)$ may lead to the same observations matrix $\bbY$. Due to its theoretical relevance (with practical consequences), this will be the subject of study of this section. Subsequently, we will propose two different approaches to solve the problem: the first based on a selection sampling scheme, the second based on subspace fitting.

\sn
%
%
\textbf{Ambiguities.} 
\label{sec:ambiguities}
 If $\bbV$ and $\bbW$ are the true matrix factors satisfying~\eqref{eq:matrix-factorization}, then for any $K \times K$ invertible matrix $\bbT$, it holds:
\begin{align}
\label{eq:model-ambiguity}
    \bbY= \bbV\bbT\bbT^{-1}\bbW= \bbV^\prime \bbW^\prime.
\end{align}
However, $\bbV^\prime$ needs to be Vandermonde in order for~\eqref{eq:model-ambiguity} to be consistent with the model structure in~\eqref{eq:matrix-factorization}. For this reason, we now investigate \textit{which class of matrices maps a Vandermonde matrix to another Vandermonde matrix}, and is thus responsible for the ambiguities in~\eqref{eq:matrix-factorization}. 

\sn
To study the possible ambiguities of the system $\bbV_1\bbT=\bbV_2$, with $\bbV_1$ and $\bbV_2$ Vandermonde matrices, means to study which linear transformation $\bbT$ transforms a Vandermonde vector with parameter $x \in \reals$ to another Vandermonde vector with parameter $y$ (independent of the specific parameter instantiation), i.e.,:
\begin{align}
\label{eq:vandermonde-to-vandermonde}
   \underbrace{[1, y, y^2, \ldots, y^{K-1}]}_{\bbv(y)^\top}= \underbrace{[1, x, x^2, \ldots, x^{K-1}]}_{\bbv(x)^\top}\bbT
\end{align}
for some $y \in \reals$.
 We show how matrix $\bbT$ needs to be an upper (generalized) Pascal matrix~\cite{pascal} with parameters $(t_0,t_1)$, denoted as $\bbT_{(t_0,t_1)}$; that is, a matrix of the form\footnote{We show $K=4$ to ease the visualization.}:
 \begin{align}
     \bbT_{(t_0,t_1)}:= \begin{bmatrix}
        1 & t_0 & t_0^2 & t_0^3 \\
        0 & t_1 & 2t_0t_1 & 3t_0^2t_1\\
        0 & 0 & t_1^2 & 3t_0t_1^2  \\
        0 & 0 & 0  & t_1^3\\
    \end{bmatrix} 
 \end{align}
\begin{thm}
\label{eq:theorem}
 The equality $\bbv(y)^\top= \bbv(x)^\top\bbT$ [cf.~\eqref{eq:vandermonde-to-vandermonde}] holds for generic $x, y$ if and only if $\bbT \in \reals^{K \times K}$ is a Pascal matrix with parameters $(t_0,t_1)$.
\end{thm}

 \begin{proof}
$\Longleftarrow )$ If $\bbT_{(t_0,t_1)}$ is Pascal with parameters $t_0$ and $t_1$, then it directly follows that $\bbv(x)^\top\bbT_{(t_0,t_1)}=\bbv(t_0+t_1x)=\bbv(y)^\top$.

\sn
$\Longrightarrow )$
Consider $\bbT(:,1)= [t_0, t_1, \ldots, t_{K-1}]^\top$ so that $y= t_0 + t_1x+ t_2x^2 + \ldots + t_{K-1}x^{K-1}$, which models the second column of $\bbv(y)^\top$ as a polynomial of order $K-1$ with coefficients given by the second column of $\bbT$. 
The last element of $\bbv(y)^\top$ is then given by: 
\begin{align}
    y^{K-1}= (t_0 + t_1x+ t_2x^2 + \ldots + t_{K-1}x^{K-1})^{K-1}
\end{align}
whose highest order term of is $t_{K-1}^{K-1}x^{(K-1)^2}$. Since  $x^{(K-1)^2}$ is not present in $\bbv(x)^\top$,  $t_{K-1}$ necessarily needs to be equal to zero. Next, all the coefficients associated to the monomial $x^j$, with $j>(K-2)(K-1)$, include $t_{K-1}$, which is thus also equal to zero. The next monomial which includes terms with a coefficient different from zero is $x^{(K-2)(K-1)}$, for which the only non-zero term is $t_{K-2}$. However, since such monomial is not present in $\bbv(x)^\top$, we conclude $t_{K-2}=0$. With the same logic one can show that only $t_0$ and $t_1$, associated to the zeroth- and first-order monomial, can be different from zero. In other words, the entry $(i,j)$ of $\bbT$ represents the coefficient associated to $x^i$ when the binomial $y(x)=t_0+t_1x$ is raised to the power $j$, with $i,j$ starting from zero. 
\end{proof}


\noindent
\textbf{Identifiability and solvability.}
One of the consequences of Theorem~\ref{eq:theorem} is that, without additional constraints, every shifted and scaled version of the groundtruth parameter $x$ (for an appropriate $\bbW$), perfectly fits the observation model~[cf.~\eqref{eq:model-ambiguity}]. Thus the model is identifiable up to this ambiguity, which is sufficient in many contexts. For instance, in the context of (eigenvalues) graph learning~\cite{natali2022dual}, a shift and scale of the graph eigenvalues maintains the same topological structure of the original graph (removing the self loops caused by the shift). As we will see, identifiability alone (which is mainly a theoretical property) does not provide us a  guarantee of solvability (meaning, global convergence).

%
%
\vspace{-.2cm}
\subsection{Selection Sampling}
When domain knowledge for the problem at hand allows to constrain the range of the input $x$, a viable option is to perform a selection sampling strategy. The idea is to discretize the admissible domain in $G$ points, with $G>N$,  according to a specified sampling pattern. This leads to a vector $\overline{\bbx}$ of \textit{candidate} solutions and the problem becomes now to select which $N$ out of the $G$ entries of $\overline{\bbx}$ best fits (together with an appropriate coefficient matrix $\bbW$) matrix $\bbY$. This is done by introducing the $N \times G$ selection matrix $\bbPhi$, whose rows are indicator vectors, i.e., $\bbPhi(i,j)=1$ if $\overline{x}_j$ is selected to represent the true $x_i$. 

\sn 
A viable optimization problem  for the fitting problem is then:
\begin{align}
\label{eq:selection}
        \min_{\bbPhi, \bbW} \|\bbY - \bbPhi \bbV(\overline{\bbx}) \bbW\|_F^2
\end{align}
where $\overline{\bbV}:=\bbV(\overline{\bbx})$ is the ``augmented'' $G \times K$ Vandermonde matrix containing the discretized domain. By substituting the pseudoinverse solution $\hat{\bbW}=  ((\bbPhi  \overline{\bbV})^\top (\bbPhi  \overline{\bbV}(\bbPhi  \overline{\bbV})^\top) \bbY = (\bbPhi \overline{\bbV})^\dagger \bbY$ into~\eqref{eq:selection} we have:
\begin{align}
    \min_{\bbPhi} \|(\bbI - \bbPhi \overline{\bbV} (\bbPhi \overline{\bbV})^\dagger) \bbY \|_F^2.
\end{align}
Because $\bbPhi$ is completely determined by an $N$-sparse binary vector $\boldsymbol{\phi} \in \reals^G$ an exhaustive search can be performed to find the optimal selection pattern, which consists in trying all the possible $ \binom{G}{N} = \frac{G!}{N!(G-N)!}$ combinations. Alternatively, to further reduce the search space, a branch and bound technique can be employed. However, both these approaches are mainly feasible for low values of $N$ and $G$, due to the combinatorial nature of the problem.

%
%

\subsection{Subspace Fitting}
\vspace{-.3cm}
In a subspace fitting approach~\cite{viberg1991sensor}, we start by considering the following optimization problem to estimate $\bbx$ and $\bbW$:
\begin{align}
\label{eq:original}
    \min_{\bbx, \bbW} \;\frac{1}{2} \|\bbY - \bbV(\bbx)\bbW\|_F^2,
\end{align}
which can be solved, for instance, with an alternating minimization approach, without guarantee of convergence to a global optimum.

\sn
To reduce the estimation effort, especially when $L \gg K$, consider the economy-size SVD of matrix $\bbY$, i.e.,  $\bbY=\bbU\bbSigma\bbZ^\top$, where $\bbU \in \reals^{N \times K}$ and $\bbZ \in \reals^{L \times K}$  are the left and right singular vectors, respectively, and $\bbSigma \in \reals^{K \times K}$ is the diagonal matrix of singular values. Since both $\bbV(\bbx)$ and  $\bbU$ represent a basis for the column space of $\bbY$, there exists a non-singular matrix $\bbS \in \reals^{K \times K}$ such that 
$\bbV=\bbU\bbS$. The subspace fitting problem reads then as:
\begin{align}
\label{eq:subspace-fitting}
    \min_{\bbx, \bbS}\;\frac{1}{2} \|\bbV(\bbx) - \bbU\bbS\|_F^2,
\end{align}
which, upon substituting the pseudoinverse solution  $\hat{\bbS}=\bbU^\dagger\bbV$ into~\eqref{eq:subspace-fitting}, can be casted as the following equivalent problem:
\begin{align}
\label{eq:equivalent-subspace-fitting}
    \min_{\bbx}\; \{ f(\bbx):=\frac{1}{2} \|\bbP \bbV(\bbx)\|_F^2\},
\end{align}
with $\bbP:=\bbI_N - \bbU\bbU^\dagger$ the orthogonal projection matrix onto the orthogonal complement  of $\bbU$.
In other words, problem~\eqref{eq:equivalent-subspace-fitting} aims to find a vector $\bbx$ such that the Vandermonde matrix $\bbV(\bbx)$ is orthogonal to the subspace spanned by the orthogonal complement of $\bbU$.

\sn
Problem~\eqref{eq:subspace-fitting} is not convex in $\bbx$ due to the polynomial degree $K$ (unless $K-1=1$, i.e., the model in~\eqref{eq:polynomial-regression} is linear). To tackle the non-convexity of the problem, we resort to sequential convex programming (SCP)~\cite{boyd2008sequential}, a local optimization method that leverages convex optimization, where the non-convex portion of the problem is modeled by convex functions that are (at least locally) accurate. As in any non-convex problem, the initial starting point plays a big role; thus, if no prior information on the variable is given, a multi-starting point approach is advisable.

\begin{remark}
\label{remark:music}
As mentioned in Remark~\ref{remark:doa}, the nature of the problem and the  formulation~\eqref{eq:equivalent-subspace-fitting} share similarities with the MUSIC algorithm~\cite{schmidt1986multiple}. However, while in MUSIC every column of matrix $\bbV(\bbx)$ (containing the steering vectors of the array manifold)  depends solely on one scalar parameter and $N$ independent $1$-dimensional searches can be carried out, here each column contains all the $N$ variables. An $N$-dimensional search is thus needed, which makes a ``scanning'' of the vector variable $\bbx$ infeasible, unless $N$ is very small.
\end{remark}

\sn
\textbf{SCP.} The general idea of SCP is to maintain, at each iteration $r$, an estimate $\bbx^{r}$ and a respective convex trust region $\ccalT^{r} \subseteq \reals^{N}$ over which we trust our solution to reside. The next solution $\bbx^{r+1}$ is then obtained by minimizing, over the defined trust region $\ccalT^{r}$, a convex approximation $\hat{f}(\cdot)$ of $f(\cdot)$ around the previous estimate $\bbx^{r}$. In our case, we define as trust region the set:
\begin{align}
\label{eq:trust-region}
    \ccalT^{r}:=\{\bbx \;  \vert \;  \|  \bbx - \bbx^{r}\|_p^2 \leq \rho(r) \}
\end{align}
where $\rho: \mathbb{Z}_{+} \to \reals_{++}$ is an iteration-dependent mapping indicating the maximum admissible length of the convex $p$-norm ball  in~\eqref{eq:trust-region}.

Next, a feasible intermediate iterate is found by:
\begin{align}
    \hat{\bbx} = \argmin_{\bbx \in \ccalT^{[r]}} f(\bbx^{r}) + \nabla_{\bbx} f(\bbx^{r})^\top (\bbx - \bbx^{r})
\end{align}
i.e., by minimizing the (convex) first-order Taylor approximation of $f(\cdot)$ around $\bbx^{r}$, with
 $\nabla_{\bbx} f(\cdot) \in \reals^N$  the gradient of function $f(\cdot)$; see Appendix~A for the explicit gradient computation.
Notice that due to the non-convexity of $f(\cdot)$ [cf.~\eqref{eq:equivalent-subspace-fitting}], its value $f( \hat{\bbx})$ at the new feasible estimate $ \hat{\bbx}$ is not guaranteed to be lower than the value at $\bbx^{r}$. Thus, to find the optimal solution at iteration $(r+1)$, we pick the best convex combination of $\bbx^{r}$ and $\hat{\bbx}$ achieving the lowest function value, i.e.,:
\begin{align}
\alpha_r^\star &= \argmin_{\alpha_r} f (\alpha_r \bbx^{r} + (1-\alpha_r)\hat{\bbx}), \quad \alpha \in (0,1) \\
    \bbx^{r+1}&=\alpha_r^\star \bbx^{r} + (1-\alpha_r^\star)\hat{\bbx}.
\end{align}
We stress that a zero fitting error of the cost function can be achieved only by the true unknown parameter and its model ambiguities, which as such represent global minima of the function and solutions of the problem. 

 \begin{figure*}[]
 \vspace{-1.4cm}
 \hspace{-2.5mm}
\begin{subfigure}{0.33\textwidth}
    \centering
    \includegraphics[scale=0.45, trim= 0.5cm 0 1cm 0.5cm, clip=true ]{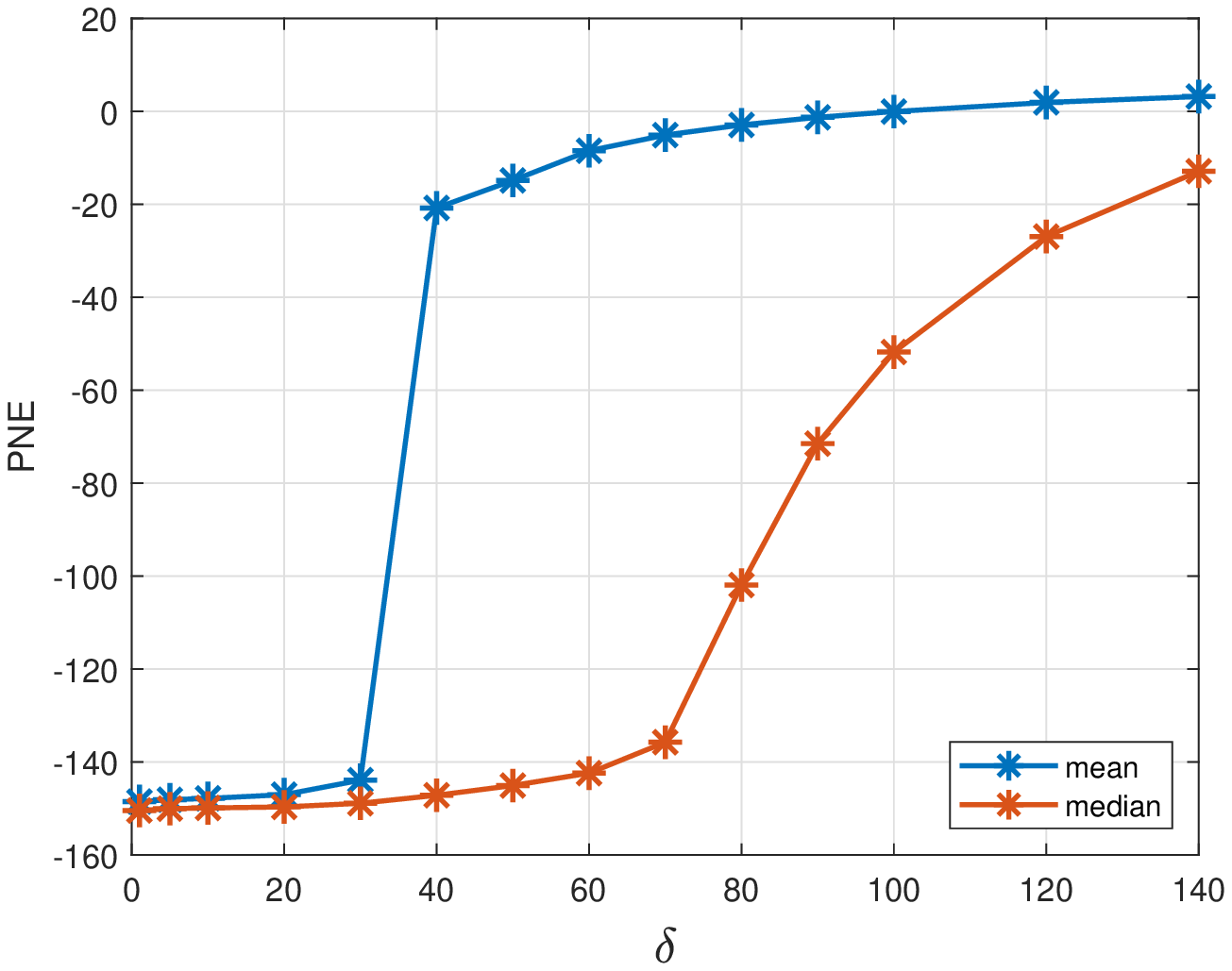}
    \label{fig:locations}
\end{subfigure}\hspace{3mm}%
\begin{subfigure}{0.33\textwidth}
    \centering
   \includegraphics[scale=0.45, trim= 0.7cm 0 1cm 0.5cm, clip=true ]{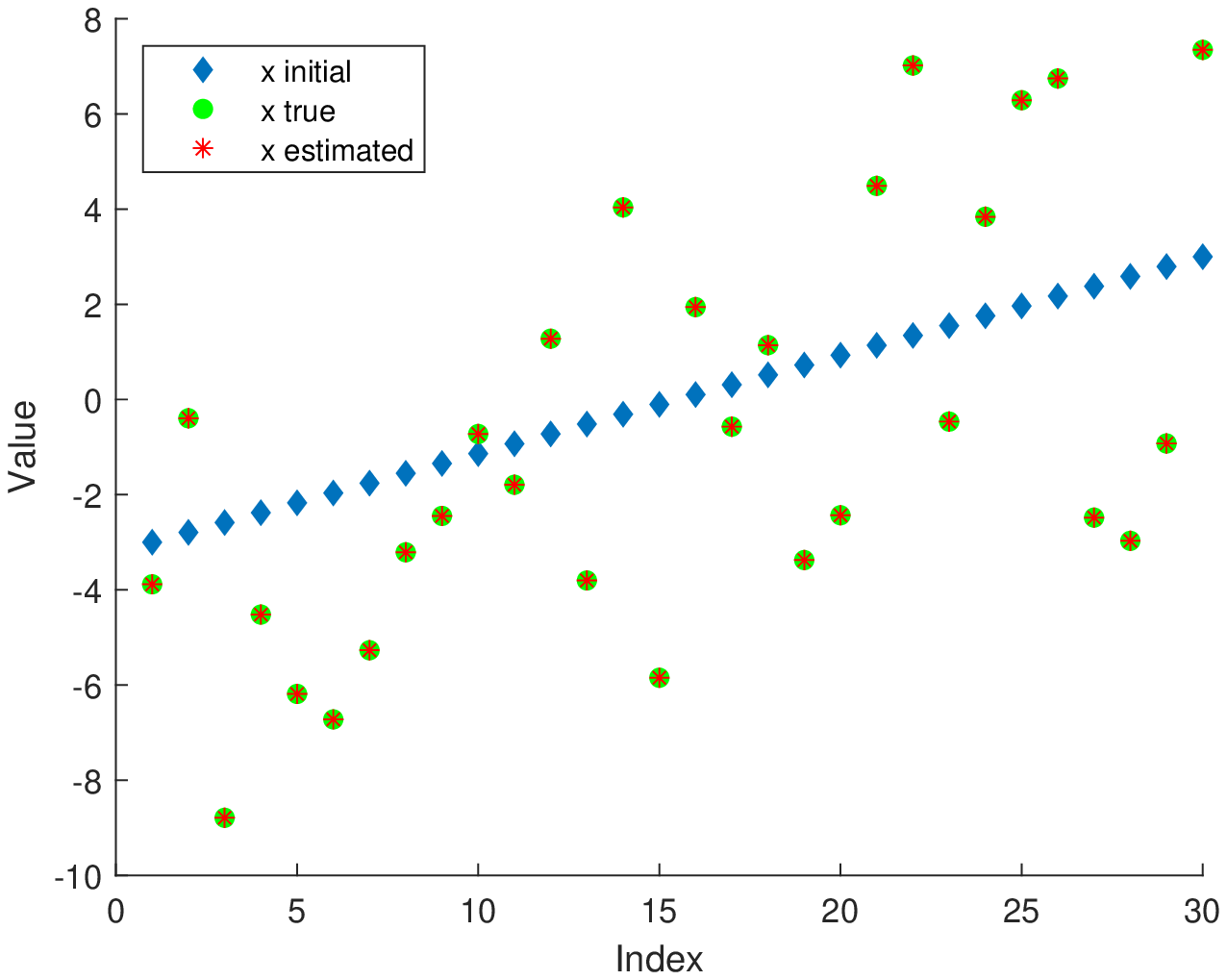}
    \label{fig:locations}
\end{subfigure}\hspace{2mm}%
\begin{subfigure}{0.33\textwidth}
\centering
\includegraphics[scale=0.45, trim= 0.7cm 0 1cm 0.5cm, clip=true]{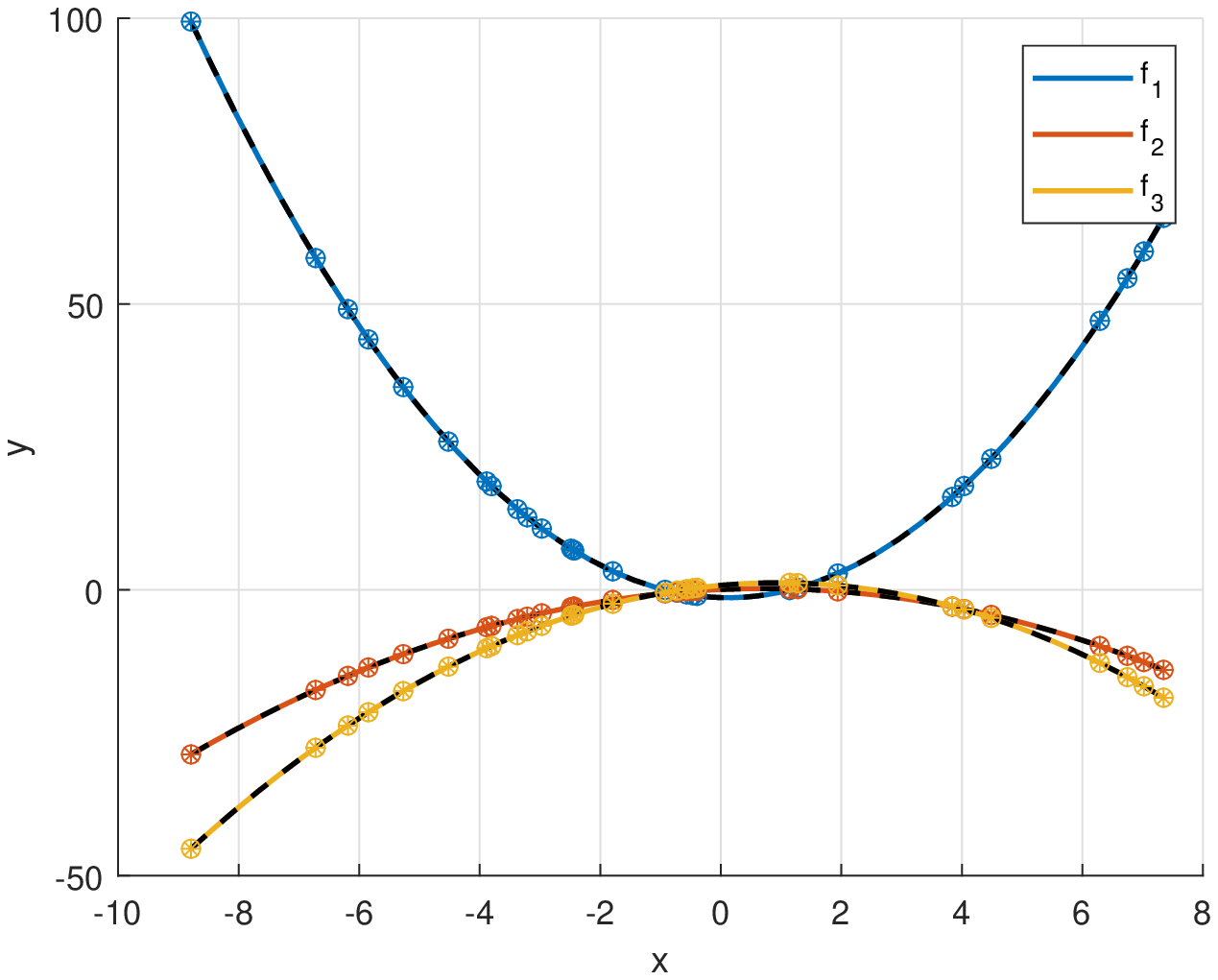}
\label{fig:polynomials}%
\end{subfigure}
\vspace{-0.6cm}
\caption{For $K=L=3$:  (left) Average and median PNE in dB over 1000 runs of the algorithm; (middle) True locations (green circles), inferred locations (red asterisks) and uniform locations (blue diamonds); the latter is used as the starting point for the algorithm; (right) Original polynomials (solid lines) with associated function samples ($\circ$) in correspondence of the original uniform locations, and \textit{inferred} polynomials (black dotted lines) with associated function samples ($*$) in correspondence of the \textit{inferred} locations. }
\label{fig:inferred}
\end{figure*}

\vspace{-.2cm}
\section{Numerical Results}
\label{sec:numerical}
 

\vspace{-.3cm}

We now perform numerical simulations to assess the validity of the subspace approach, with focus on a practical case-study: clock-synchronization after sampling jitter. 

\sn 
\textbf{Scenario.} In the context of signal sampling, jitter is an undesirable effect caused by, for instance, interference of the clock signal, and consisting of a deviation from the true sampling pattern which we consider here to be periodic. When the time between samples varies an instantaneous signal error arises and  a jitter-correction (clock-recovery) is desirable.
To this extent, consider the time-interval  $[u_0, u_1, \ldots, u_{N-1}]$ of uniformly sampled  locations (timings) with period $T$, i.e., such that $u_n - u_{n-1}= T$ for $n=1, \ldots, N-1$. Because of jitter, the actual sampling locations are given by $x_n= u_n + j_n$, with $j_n \sim \ccalN(0, \delta T/2)$ a Gaussian distribution truncated at one standard deviation. The   parameter $\delta$ specifies whether an overlap among adjacent samples $x_n$ is possible ($\delta \geq 1$) or not ($\delta < 1$).
 Associated to each $x_n$ is a signal value $y_n$ corresponding to the sampling of an unknown function $f(x_n)$.
 
 \sn 
 \textbf{Metrics.}
To assess the validity of the proposed approach, recall that  we can recover the solution of problem~\eqref{eq:equivalent-subspace-fitting} up to a shift and scaling of the true positions [cf. Section~\ref{sec:ambiguities}]. Thus, we use  the normalized error modulo Pascal (PNE) as performance metric, defined as:
%
\begin{align}
\label{eq:pascal-error}
    \text{PNE}(\hat{\bbx}, \bbx)= \min_{t_0, t_1} \frac{\| \bbx - (t_0 \boldsymbol{1} + t_1 \hat{\bbx})\|_2}{NT}
\end{align}
  which measures how far the true locations are from a linear transformation of the recovered estimates.
 Clearly~\eqref{eq:pascal-error} is zero whenever $\hat{\bbx}$ is a solution for~\eqref{eq:equivalent-subspace-fitting}.

 \sn
 \textbf{Results.}
For a fixed number $N=30$ of locations, we run  the algorithm $1000$ times to randomize the coefficient matrix $\bbW$ used to generate the polynomials, for different values of the polynomial order $K=\{3,4\}$, number of polynomials $L=\{3,4\}$, $L\geq K$, and parameter $\delta \in (1, 140)$. We focus our attention on quadratic ($K=3$) and cubic ($K=4$) polynomials because for higher degrees the polynomials are almost flat in the domain of interest and they tend to be highly sensitive to round-off due to the very ill-conditioning of the Vandermonde matrix. As function domain we consider $[-3,3]$, and discretize it with a sampling period of $T= 6/(N-1)$. We initialize the starting point of the algorithm with the uniform sampled domain, i.e., $\bbx^0=[u_0, \ldots, u_{N-1}]^\top$.

\sn In Fig.~\ref{fig:inferred} (left) is shown (in dB) the average and median PNE over $1000$ simulations for $K=L=3$ and $\delta$ ranging from $1$ to $140$. The first remarkable result is that the algorithm is able to perfectly recover the solution (up to a shift and scaling) without ending in a local minimum for very high values of $\delta$, resulting in big overlaps among the samples. In particular, both the mean and the median are zero up to $\delta=30$, which is the setting for which one sample can overlap with $15$ of its nearby samples. The median, which is more robust than the mean in terms of sensitivity to outliers, is basically zero up to $\delta=60$, i.e., the setting for which a sample can overlap with all the others, rendering it almost a random configuration of points. Nonetheless, also for the extreme case of $\delta=140$, the median still achieves an error of $1/100$th of a sampling period, which is an acceptable error.
To enjoy the algorithm's performance, in Fig.~\ref{fig:inferred} (middle) we show for $K=L=3$ and $\delta=60$,  the estimated locations (red, after ambiguity correction), together with the true locations (green) and the initial starting point of the algorithm (blue). Correspondingly, in  Fig.~\ref{fig:inferred} (right) we show the original and inferred polynomials, constructed as $\hat{\bbY}= \bbV(\hat{\bbx})\bbV(\hat{\bbx})^\dagger \bbY$.



\vspace{-.3cm}
\section{Conclusion}
\vspace{-.2cm}
In this work, we studied the problem of polynomial regression when the coefficients and the regressor are both unknown. We have shown how the model is identifiable up to a generalized Pascal transformation, responsible to scale and shift the input parameter, which is sufficient in many recovery problems. We then proposed two algorithmic routines to tackle the problems, namely an exhaustive search approach and a subspace fitting approach. Finally, we performed numerical simulations in the context of sampling jitter and show the effectiveness of the subspace fitting approach to recover the original sample locations. Interesting future research directions include: \textit{i)} studying the influence of input perturbation ($\delta$) and polynomial order ($K$) on the reconstruction error for jitter scenarios; \textit{ii)} understanding  the recoverability property in case of noise in the output; \textit{iii)} studying the problem under a purely (structured) matrix factorization approach with a Vandermonde factor.
\vspace{-.25cm}

\appendix
\section{Gradient Computation}
\vspace{-.2cm}

Consider the function $f(\bbx)= \|\bbP\bbV(\bbx)\|_F^2$ and define the composite matrix function  $\bbM(\bbV(\bbx))= \bbP\bbV(\bbx)$. We use the chain rule and to compute the gradient of $f(\cdot)$ w.r.t. $\bbx$, i.e., $\partial f (\cdot)/ \partial \bbx $.
We have $ f(\bbM)= \frac{1}{2} \tr(\bbM^\top \bbM)= \frac{1}{2} \bbM : \bbM$. The differential of $f(\cdot)$ is $df(\bbM, d\bbM)= \bbM : d\bbM= \bbM : \bbP d\bbV(\bbx)$. For the differential $d\bbV(\bbx, d\bbx)$ notice that, by selecting one column, e.g., $\bbx^3$, we have $(\bbx + d\bbx)^3= \bbx^3 + 3 \bbx^2 d\bbx + \ccalO(d\bbx^2)$, and more generally:
\begin{align}
 d\bbV (\bbx) &= [\boldsymbol{0} \; \boldsymbol{1} \ldots (K-1)\bbx^{K-2}]  \odot [d\bbx \;  d\bbx \ldots d\bbx]\\
&= (\bbV(\bbx) \bbD_k) \odot (d\bbx \boldsymbol{1}^\top)= \Diag(d\bbx)(\bbV(\bbx) \bbD_k)
\end{align}
where $\bbD_k= \operatorname{SupDiag}(1, \ldots, K-1) \in \reals^{K \times K}$ is the differentiation matrix for polynomials and $\odot$ is the element-wise product.
%
By plugging this differential into the previous expression for $df$ we obtain:
\begin{align*}
    df(\bbM,d\bbM)= \bbM :  \bbP d\bbV(\bbx) =\bbM :  \bbP  \Diag(d\bbx)(\bbV(\bbx) \bbD_k)  \\
     =\bbP^\top \bbM (\bbV(\bbx) \bbD_k)^\top :    \Diag(d\bbx)  
     =\diag(\bbP^\top \bbM (\bbV(\bbx) \bbD_k)^\top) :    d\bbx 
\end{align*}
We conclude $ \partial f(\bbx)/ \partial \bbx= \diag(\bbP^\top \bbP\bbV(\bbx)\bbD_k^\top \bbV(\bbx)^\top) $.

\newpage

\bibliographystyle{IEEEbib}
\bibliography{refs}

\end{document}